\documentclass[letterpaper, 10 pt, conference]{ieeeconf}  

\IEEEoverridecommandlockouts                              %

\usepackage{amssymb, amsmath, amsfonts, bbm, empheq, graphicx, float, mathrsfs}

\let\labelindent\relax %
\usepackage{enumitem}

\usepackage[dvipsnames]{xcolor}

\usepackage[hyphens]{url}
\usepackage[normalem]{ulem} %

\usepackage[linktocpage=true, colorlinks=true, allcolors=PineGreen,pdfusetitle]{hyperref}

\usepackage[capitalise]{cleveref} %
\usepackage{autonum} %

\newtheorem{theorem}{Theorem}[section] %
\newtheorem{lemma}[theorem]{Lemma}

\newtheorem{proposition}[theorem]{Proposition}

\crefname{lemma}{Lemma}{Lemmas}
\crefname{definition}{Definition}{Definitions}
\crefname{proposition}{Proposition}{Propositions}
\crefname{corollary}{Corollary}{Corollaries}
\crefname{example}{Example}{Examples}

\usepackage{algorithm}
\usepackage{algpseudocode}
\algrenewcommand\algorithmicrequire{\textbf{Input:}}
\algrenewcommand\algorithmicensure{\textbf{Output:}}

\newcommand\E{\mathbb{E}} 
\newcommand\R{\mathbb{R}}

\newcommand\Z{\mathbb{Z}}
\newcommand\1{\mathds{1}}
\newcommand\N{\mathds{N}}
\def\ddefloop#1{\ifx\ddefloop#1\else\ddef{#1}\expandafter\ddefloop\fi}
\def\ddef#1{\expandafter\def\csname b#1\endcsname{\ensuremath{\boldsymbol{#1}}}}
\ddefloop ABCDEFGHIJKLMNOPQRSTUVWXYZ\ddefloop

\def\ddef#1{\expandafter\def\csname #1hat\endcsname{\ensuremath{\widehat{\csname #1\endcsname}}}}
\ddefloop {alpha}{beta}{gamma}{delta}{epsilon}{varepsilon}{zeta}{eta}{theta}{vartheta}{iota}{kappa}{lambda}{mu}{nu}{xi}{pi}{varpi}{rho}{varrho}{sigma}{varsigma}{tau}{upsilon}{phi}{varphi}{chi}{psi}{omega}{Gamma}{Delta}{Theta}{Lambda}{Xi}{Pi}{Sigma}{varSigma}{Upsilon}{Phi}{Psi}{Omega}{ell}\ddefloop

\def\ddef#1{\expandafter\def\csname #1scr\endcsname{\ensuremath{\mathcal{#1}}}}
\ddefloop ABCDEFGHIJKLMNOPQRSTUVWXYZ\ddefloop

\def\ddef#1{\expandafter\def\csname #1cal\endcsname{\ensuremath{\mathscr{#1}}}}
\ddefloop ABCDEFGHIJKLMNOPQRSTUVWXYZ\ddefloop

\def\ddef#1{\expandafter\def\csname #1hat\endcsname{\ensuremath{\widehat{#1}}}}
\ddefloop ABCDEFGHIJKLMNOPQRSTUVWXYZ\ddefloop

\def\ddef#1{\expandafter\def\csname #1hat\endcsname{\ensuremath{\widehat{#1}}}}
\ddefloop abcdefghijklmnopqrsuwvyxz\ddefloop

\def\ddef#1{\expandafter\def\csname #1bar\endcsname{\ensuremath{\bar{#1}}}}
\ddefloop ABCDEFGHIJKLMNOPQRSTUVWXYZ\ddefloop

\def\ddef#1{\expandafter\def\csname #1bar\endcsname{\ensuremath{\bar{#1}}}}
\ddefloop abcdefghijklmnopqrstuvwxyz\ddefloop

\def\ddef#1{\expandafter\def\csname #1bb\endcsname{\ensuremath{\mathbb{#1}}}}
\ddefloop ABCDEFGHIJKLMNOPQRSTUVWXYZ\ddefloop

\newcommand{\xpln}[1]{, \ \text{#1}} %

\DeclareMathOperator*{\amin}{arg\,min}

\DeclarePairedDelimiter\absBase{\lvert}{\rvert}
\newcommand\abs[1]{\absBase*{#1}}

\DeclarePairedDelimiterXPP{\nucnormBase}[1]{}\lVert\rVert{_\mathrm{nuc}}{#1} %

\DeclarePairedDelimiterXPP{\opnormBase}[1]{}\lVert\rVert{_2}{#1} %
\newcommand\opnorm[1]{\opnormBase*{#1}}
\DeclarePairedDelimiterXPP{\FnormBase}[1]{}\lVert\rVert{_{\mathrm{F}}}{#1} %
\newcommand\Fnorm[1]{\FnormBase*{#1}}
\DeclarePairedDelimiterXPP{\enormBase}[1]{}\lVert\rVert{_2}{#1} %
\newcommand\enorm[1]{\enormBase*{#1}}
\DeclarePairedDelimiter\setBase{\{}{\}}
\newcommand\set[1]{\setBase*{#1}}

\newcommand\matb[1]{\begin{bmatrix} #1 \end{bmatrix}}

\DeclareMathOperator{\Tr}{Tr}

\DeclarePairedDelimiter\lrpBase{(}{)}
\newcommand\lrp[1]{\lrpBase*{#1}}
\DeclarePairedDelimiter\lrbBase{[}{]}
\newcommand\lrb[1]{\lrbBase*{#1}}
\DeclarePairedDelimiter\brkBase{\langle}{\rangle} %
\newcommand\brk[1]{\brkBase*{#1}}
\DeclarePairedDelimiterXPP\ExpBase[1]{\E}{[}{]}{}{#1} %

\let\Pr\relax
\DeclarePairedDelimiterXPP\PrBase[1]{\P}{[}{]}{}{#1}
\newcommand\Pr[1]{\PrBase*{#1}}

\DeclarePairedDelimiter\floor{\lfloor}{\rfloor}

\newcommand{\lsim}{\lesssim}
\newcommand{\gsim}{\gtrsim}
\newcommand{\eps}{\varepsilon}
\newcommand{\by}{\times}

\renewcommand*{\P}{\mathbb{P}}

\newcommand{\pseudo}[1]{ {#1}^\dagger}
\newcommand{\inv}[1]{{#1}^{-1}}

\newcommand{\pU}{\pseudo{U}}

\renewcommand\u[1]{^{(#1)}}
\newcommand\uL{\u{L}}

\usepackage{dsfont}

\newcommand\Tf{T_{\mathrm{eff}}}
\renewcommand\Tbar{T} %

\newcommand\Ntot{N_{\mathrm{tot}}}

\newcommand\bigM{\mathscr{M}}

\title{\LARGE \bf
Learning Clusters of Partially Observed Linear Dynamical Systems
}
\author{Maryann Rui and Munther A. Dahleh%
\thanks{This work is an updated and extended version of a paper presented at the 2025 American Control Conference (ACC).}%
\thanks{M. Rui and M. A. Dahleh are with the Department of Electrical Engineering and Computer Science, Massachusetts Institute of Technology.
        {\tt\small mrui@mit.edu, dahleh@mit.edu}}%
}

\begin{document}

\maketitle
\thispagestyle{plain}
\pagestyle{plain}

\begin{abstract}
We study the problem of learning clusters of partially observed linear dynamical systems from multiple input-output trajectories. This setting is particularly relevant when there are limited observations (e.g., short trajectories) from individual data sources, making direct estimation challenging. In such cases, incorporating data from multiple related sources can improve learning.
We propose an estimation algorithm that leverages different data requirements for the tasks of clustering and system identification. First, short impulse responses are estimated from individual trajectories and clustered. Then, refined models for each cluster are jointly estimated using multiple trajectories. 
We establish end-to-end finite sample guarantees for estimating Markov parameters and state space realizations and highlight trade-offs among the number of observed systems, the trajectory lengths, and the complexity of the underlying models.
\end{abstract}

\section{Introduction}
Learning dynamical systems is a fundamental problem in a wide range of fields, such as robotics, healthcare, and finance. 
When data is generated from multiple related systems that have a shared structure, it can be advantageous to leverage data across the multiple sources to learn the common structure and to improve estimates of constituent systems. 
For instance, different machines may operate in only three distinct modes, or the response of a health indicator may fall into one of just five characteristic regimes, and these underlying models may be jointly learned. 
In this paper, we focus on the versatile setting of clusters of input-output models of partially observed linear systems, and study how our ability to estimate systems can improve with not only the length of observed trajectories but also the number of systems available for learning.
In particular, we provide and analyze an algorithm for learning clusters of linear systems that distinguishes between the amount of data needed, in terms of trajectory lengths, to cluster systems versus to estimate state space realizations of the systems. 

\subsection{Model}
To begin, suppose we observe trajectories from $\Ntot$ linear time-invariant (LTI) systems. Each system $i\in [\Ntot]$ has a finite order minimal realization $M_i = (C_i, A_i, B_i, D_i)$, with dynamics given by
\begin{align}\label{eq:st-sp-model}
x_{i, t+1} &= A_i x_{i,t} + B_i u_{i,t} + w^{(1)}_{i,t} \\
y_{i,t} &= C_i x_{i,t} + D_i u_{i,t} + w^{(2)}_{i,t}
\end{align}
for all $t \in \Z$. At each time $t$, $x_{i,t} \in \R^n$ is the state, $u_{i,t}\in \R^m$ is the input, $y_{i,t} \in \R^p$ is the output, and $w^{(1)}_{it}\in \R^n$ and $w^{(2)}_{it} \in \R^p$ are the process noise and measurement noise, respectively.  $A_i \in \R^{n \by n}$ is the state transition matrix, $B_i \in \R^{m\by n}$ the control matrix, $C_i\in \R^{p\by n}$ the measurement matrix, and $D_i \in \R^{p\by m}$ the feedthrough matrix, with $n$ as the state dimension, $m$ the input dimension, and $p$ the output dimension. 
When we write that a system or model is given by a realization $(C, A, B, D)$, that system corresponds to the equivalence class generated by the similarity transformations $(CQ^{-1}, QAQ^{-1}, QB, D)$ of the realization, for any invertible $n\by n$ matrix $Q$. 
For each system $i \in [\Ntot]$, we observe a length $\Tbar$ trajectory of inputs and outputs, $\set{u_{i,0}, (u_{i,t}, y_{i,t}) \mid t = 1, \dots, \Tbar}$.

\paragraph{Clusters of Linear Systems}
We now introduce the latent clustered structure shared across systems. 
Define the distance between two models $M_1$ and $M_2$ in terms of the distance between their impulse response sequences, $
	d(M_1, M_2)^2 = \sum_{t=0}^{\infty} \Fnorm{C_1A_1^tB_1 - C_2 A_2^t B_2}^2$,
where $\Fnorm{\cdot}$ is the Frobenius norm.

In our cluster model, we assume that the model $M_i$ of each system $i \in [\Ntot]$ can be well-approximated by one of $K$ different LTI models, the set of which we denote $\bigM = \set{M_{(k)}=(C_{(k)}, A_{(k)}, B_{(k)}, D_{(k)}) : k \in [K]}$. 
That is, if $\kappa_i\in [K]$ denotes the true cluster assignment of each system $i\in [\Ntot]$, when $\kappa_i = k$, the distance $d(M_1, M_{(k)})<\delta_{(k)}$, where $\delta_{(k)}\geq 0$ is small and represents the cluster width. 
While our analysis focuses on the case of $\delta_{(k)}  = 0$, which is similar to a mixture model, we demonstrate the robustness of the algorithm to nonzero cluster widths in simulations. %
Note we use the term ``system'' to refer to each of the $\Ntot$ sources of data, and the terms ``model'' or ``cluster'' one of the $K$ underlying dynamics that describe each system.

\paragraph{Problem Statement}
Given input-output data $\set{(u_{i,0}, (u_{i,t}, y_{i,t}) \mid i \in [\Ntot], t \in [T])}$, we wish to learn either impulse response or state space models that capture the dynamics generating each observed trajectory $i \in [\Ntot]$. Assuming our cluster model, we content ourselves with estimating the $K$ models in $\bigM$ and the cluster assignments $\kappa_i$ for each $i \in [\Ntot]$. 
\subsection{Overview of Contributions}
Our goal is to estimate state space models for each observed system, while taking advantage of the shared structure among the systems, namely, that there are $K$ underlying models, to improve our estimation. Of particular interest are settings in which we have many individual systems, but only observe a short trajectory from each system. How can we leverage the ``spatial'' information, i.e., information shared across systems, to help estimate models of individual systems despite limited ``temporal'' information, i.e., short lengths of individual trajectories? 

To answer this question, we propose a simple method to learn models for clusters of linear systems, with the following steps. First, estimate a small number $L_1$ of Markov parameters for each system individually with a standard least squares estimator. Second, cluster systems based on these estimates using $k$-means or other distance-based clustering algorithms. Third, calculate a refined model estimate for each cluster by estimating $L_2$ Markov parameters from trajectories within the cluster. Here, $L_2$ may be greater than $L_1$, and not all systems in a given cluster are required to have trajectories of length $L_2$. Finally, use the Ho-Kalman algorithm to obtain state space realizations.

This multi-step approach allows us to decouple the number of observations needed per system in order to determine its cluster membership, based on the first few estimated Markov parameters, and the number of observations needed per system in order to estimate a state space realization, which could be much greater, scaling with the system order $n$. As long as every cluster contains a few trajectories which are sufficiently long to estimate a state space realization, systems with shorter trajectories can also be identified simply by mapping them to the correct cluster. 

Additionally, refined estimation using multiple trajectories within a cluster improves the estimation error of Markov parameters by a factor of $1/\sqrt{N}$ over single trajectory estimation, where $N$ is the number of trajectories in that cluster. Thus we can quantify the tradeoff between the number of trajectories $N$ and the length of each trajectory $\Tbar$ needed to estimate Markov parameters to within a given accuracy.

Our contributions are summarized as follows. 
\begin{itemize}
\item We propose a method (\cref{alg:sysid}) for learning clusters of partially observed linear systems. 
\item We derive end-to-end finite sample guarantees for estimating Markov parameters and state space realizations for clusters of stable linear systems (Proposition \ref{prop:end2end}).
\item As a key step, in Theorem \ref{thm:MAIN-RESULT}, we prove finite-sample error bounds for estimating Markov parameters from multiple trajectories. In particular, we quantify the tradeoff in the number of trajectories and the length per trajectory needed to estimate Markov parameters to within a given accuracy. 
\item We empirically evaluate the algorithm on simulated linear system mixtures, varying $N$, $T$, and cluster widths, and compare its performance with a moment-based estimator that uses tensor decomposition. %
\end{itemize}

\subsection{Related Work}
Our work is related to the literature on estimating multiple systems with a shared latent structure, nonparametric clustering of time series,
and finite sample error bounds in system identification. 
\subsubsection{Latent shared structure}
\paragraph{Fully observed systems} 
A variety of recent works study the estimation of multiple dynamical systems with various shared structures, but are usually restricted to the setting of fully observed models of the form
\begin{align}\label{eq:state-dyn-2}
x_{i,t+1} = A_i x_{i,t} +B_i u_{i,t} + w_{i,t} \xpln{$ i \in [N]$}
\end{align}
with state and input observations $\{x_{i,t}, u_{i,t}\}_{i \in [N], t \in [T]}$ and noise $w_{i,t}$, with various assumptions on the shared structure of the state transition matrices $\set{A_i}$ (and when relevant, the input matrices $\set{B_i}$). 
\cite{modi2021joint,zhang2023multi} study low rank settings via empirical risk minimization. In a similar setting, \cite{rui2023estimation} proposes an algorithm based on least squares estimation and SVD, though the analysis is restricted to static linear models. 
Mixtures or cluster models for \eqref{eq:state-dyn-2} \cite{chen2022learning, toso2023learning}, and relatedly, models of bounded heterogeneity conditions \cite{xin2023learning, chen2023multi}, where $\opnorm{A_i - A_j} \leq \eps$ for $i, j \in [N]$, have been explored, with solutions based on regularized least squares.

Note however, that bounded heterogeneity of system parameters is not system-theoretic in the sense that it does not guarantee similar system behavior.
Two systems may be close in the sense that $\opnorm{A_1 - A_2} \leq \eps$, but $A_1$ may be strictly stable while $A_2$ is unstable. For sufficiently long trajectories, the difference in system behavior can be dramatic. In our framework, we define our cluster distances to be in terms of Markov parameters, i.e., impulse responses, which yields a more appropriate measure of system similarity. 

Furthermore, the aforementioned structural assumptions on $\set{(A_i, B_i)}_{i \in [N]}$ do not readily extend to structures on the matrix representations $(C, A, B, D)$ of partially observed systems, since such sets of matrices are only meaningful up to similarity transformations. This motivates our definition of framework of clustering systems by impulse response.

\paragraph{Time series clustering} 
Some works cluster time series using model-free approaches such as distance metrics on raw or transformed trajectories \cite{javed2020benchmark, aghabozorgi2015time}. We focus on model-based clustering to learn underlying the dynamical models of the systems, not just cluster observations.

\paragraph{Partially observed systems} 
The clustered setting is very similar to the setting of mixture models of dynamical systems, where each observed trajectory is generated from a randomly selected model $M_{(k)}$ following a multinomial distribution. %
Closest to our work are \cite{bakshi2023tensor, rui2024finite}, which use tensor decomposition of 6th-order moments to estimate mixture parameters. They propose different estimators with different sample complexities, but are sensitive to data conditioning and may require substantial data for higher order moments. 
Additionally, tensor approaches assume a linear independence, or non-degeneracy, condition on the mixture Markov parameters, in contrast to our cluster separation condition. Non-degeneracy requires the number of clusters $K$ to be less than the dimension of the Markov parameters to be estimated, while cluster separation does not. In general, these assumptions are not directly comparable.

\subsubsection{Finite sample bounds for system identification}
Finite sample error bounds in system identification have been extensively studied, often for single trajectories. Earlier works (e.g., \cite{campi2002finite}) focused on prediction error bounds for estimated linear models, represented by rational transfer functions. Recent works \cite{sarkar2021finite, oymak2021revisiting} provide bounds for learning strictly stable systems using least squares estimates of Markov parameters or of a Hankel matrix directly, followed by the Ho-Kalman algorithm \cite{ho1966effective} to obtain a system realization.
In \cite{bakshi2023new}, a moment-based estimator of the Markov parameters is proposed for learning stable and marginally stable systems from single trajectories, with the estimation error bounded polynomially in system parameters, though the explicit rates of dependence are not provided.

A core part of our analysis is in deriving finite-sample error bounds for estimating Markov parameters from multiple independent trajectories. While \cite{zheng2020non} gives $O(1/\sqrt{N})$ bounds for stable and unstable systems, their error bounds grows superlinearly with trajectory length $T$, making them unsuitable for strictly stable systems. In contrast, we show that for strictly stable systems, the error decreases as $1/\sqrt{NT}$ for $N$ trajectories of length $T$.

Finally, \cite{fattahi2021learning} uses $\ell_1$-regularization to improve the sample complexity of estimating $L$ Markov parameters from a single length $T$ trajectory to be poly-logarithmic in the system dimensions, but decaying as $O(1/T^{1/4})$. In certain high-dimension, low-data regimes, the $\ell_1$-regularized estimator can outperform the standard least squares estimator. 
Incorporating their estimator to further reduce data requirements per individual systems would be an interesting next step. 

\section{Setup} 
\subsection{Notation}
For any natural number $N \in \N$ we define the sets $[N]\coloneq \set{1, 2, ..., N}$ and $[0{:}N] \coloneq \set{0, 1, ..., N}$.
For a matrix $A$, $\Tr(A)$, $A'$, and $\pseudo{A}$ denote its trace, transpose, and Moore-Penrose pseudoinverse, respectively. $\Fnorm{A}$ and $\opnorm{A}$ denote the Frobenius norm and the operator, or spectral, norm of matrix $A$, respectively. The identity matrix in $\R^{d\by d}$ is written $I_d$. 
For an $n \by m$ matrix $M$, $\sigma_k(M)$ denotes the $k$th largest singular value of $M$, and $\sigma_{\min}(M)\coloneq \sigma_{\min(n,m)}(M)$ and $\sigma_{\max}(M)\coloneq\sigma_1(M)$ the smallest and largest singular values of $M$, respectively. 
For a square matrix $A$, $\rho(A)$ denotes its spectral radius.  %
We use $c$ to denote a universal positive constant, which may change in value between lines. For any real-valued functions $a, b$, the inequality $a \lsim b$ means that $a \leq cb$ for some $c$. 
Unless otherwise specified, all random variables are defined on the same probability space. 

For a model $M$ with a realization $(C, A, B, D)$, the $k$th Markov, or impulse response, parameter is given by $G_k = CA^{k-1}B$ for $k \geq 1$ and $G_0 = D$ for $k = 0$. Note that the Markov parameters are invariant under similarity transformations of the system matrices.
We also define the following Gramian-like quantities for $M$: 
\begin{align}
\Gamma_\infty(M) &\coloneq \sum_{t=0}^\infty \sigma_u^2 A^t B(A^t B)' + \sigma_w^2 A^t (A')^t \in \R^{n \by n}\\
\Gamma^O_\infty(M) &\coloneq  I_p + \sum_{t=0}^\infty CA^t(CA^t)' \in \R^{p \by p}.
\end{align}

\subsection{Assumptions}\label{sec:assumptions} %
\paragraph{Dynamics} We assume each model $k \in [K]$ has a minimal realization $(C_{(k)}, A_{(k)}, B_{(k)}, D_{(k)})$, ensuring controllability and observability, with a common, known system order $n$. We consider strictly stable systems with $\rho(A_{(k)}) < 1$, and a zero initial condition $x_{i,0} = 0$ for $i \in [\Ntot]$. 
\paragraph{Distributional assumptions} For $i \in [\Ntot], t\geq 0$, the inputs $u_{i,t}$ are independent zero-mean isotropic Gaussian random vectors with variances bounded by $\sigma_u^2 I_m$. As is standard \cite{oymak2021revisiting}, we assume the process noise $w^{(1)}_{i,t}$ and measurement noise $w^{(2)}_{i,t}$ are mutually independent and independent of the inputs $u_{i,t}$, and are zero-mean subgaussian random vectors with variance proxies bounded by  $\sigma_{w^{(1)}}^2$ and $\sigma_{w^{(2)}}^2$, respectively. Let $\sigma_w = \max(\sigma_{w^{(1)}}, \sigma_{w^{(2)}})$.

\paragraph{Cluster assumptions}
We assume that we observe an equal number $N$ of trajectories from each model cluster in $\bigM$, so that $\Ntot = KN$, where the number of models, $K$, is known. For our theoretical analysis, we assume the setting of exact (zero-width) clusters, where trajectories are generated exactly from one of the $K$ models. In simulations, we show how the width of the clusters, i.e., the spacing $d(M_1, M_2)$ between two systems $M_1$ and $M_2$ assigned to the same cluster $M_{(k)}$, affects the cluster estimation accuracy. 

In order to guarantee successful clustering systems based on the first $L + 1$ Markov parameters, we assume that the $K$ clusters are well-separated in impulse response:
\begin{align}\label{eq:cluster-distance}
	\sum_{t=0}^L \Fnorm{C_{(i)} A_{(i)}^t B_{(i)} - C_{(j)} A_{(j)}^t B_{(j)}}^2 > \Delta^2,
\end{align}
so that the distance between any two models $M_{(i)}, M_{(j)} \in \bigM$, with $i \neq j$, is lower bounded. The minimum separation between clusters, $\Delta$, required for successful clustering will depend on $L, \Tbar$, and system parameters, as quantified in Proposition \ref{prop:end2end}.
Note that for strictly stable systems, meaningful differences are captured in the first $L+1$ Markov parameters for sufficiently large $L$, as the Markov parameters decay exponentially.

\subsection{System Trajectories}
Given a single model $(C, A, B, D)$ for which we observe $N$ independent trajectories of length $\Tbar$, for any $t \in [T]$ and $0 \leq L \leq t$,
the output of system $i \in [N]$ at time $t$ is given by is given by
\begin{align} %
y_{i,t} &= CA^{L}x_{i,t-L} + \sum_{j=1}^L CA^{j-1} \lrp{Bu_{i,t-j} + w^{(1)}_{i,t-j}}\\
&\quad  + Du_{i,t} + w^{(2)}_{i,t}.
\end{align}

Define the vectors of inputs and process and measurement noise $\ubar_{i,t}^{(L)} \coloneq \begin{bmatrix} u_{i,t} & u_{i,t-1} & \cdots & u_{i,t-L} \end{bmatrix}'$ and 
$\wbar_{i,t}^{(L)} \coloneq  \big[w_{i,t}^{(2)} \quad w_{i,t-1}^{(1)} \quad  \cdots \quad w_{i,t-L}^{(1)}\big]'$,
as well as the matrices
\begin{align}\label{eq:GF-defn}
G^{(L)} &\coloneq \matb{D & CB & CAB & \cdots & CA^{L-1}B}_{p\times m(L+1)}, \\
F^{(L)} &\coloneq \matb{I_p & C & CA & \cdots & CA^{L-1}}_{p\times (p+Ln)},
\end{align}
Then we may write the equation for the output in vector form as 
\begin{align}\label{eq:vector-form}
y_{i,t} = G^{(L)} \ubar_{i,t}^{(L)} + F^{(L)} \wbar_{i,t}^{(L)} + CA^Lx_{i,t-L}, \ 0 \leq L \leq t,
\end{align}
where the last term $CA^Lx_{i,t-L}$ is a residual error term due to the unobserved state $L$ time steps ago. 
Note that $G\uL$ is a block matrix consisting of the first $L+1$ Markov parameters of the system, $\set{G_0, \ldots, G_{L}}$.

We define the matrices 
\begin{align}
Y\uL_i &= \matb{y_{i,L} & y_{i,L+1} & \cdots & y_{i,\Tbar}}, \\
U\uL_i &= \matb{\ubar\uL_{i,L} & \ubar\uL_{i, L+1} & \cdots&  \ubar\uL_{i, \Tbar}}, \\ 
E\uL_i &= CA^L \matb{x_{i,0} & x_{i,1} & \cdots & x_{i,\Tbar - L}}, \\
W\uL_i &= \matb{\wbar\uL_{i,L} & \wbar\uL_{i, L+1} & \cdots & \wbar\uL_{i,\Tbar}},
\end{align}
and concatenate them as
\begin{align}\label{eq:matrix-defns-std}
\bY\uL &=\matb{Y\uL_1 & \cdots & Y\uL_N}_{p\by N\Tf},\\
 \bU\uL & = \matb{U\uL_1 & \cdots & U\uL_N}_{m(L+1)\by N\Tf}, \\ 
\bE\uL &= \matb{E\uL_1 & \cdots & E\uL_N}_{p\by N\Tf}, \\
\bW\uL &= \matb{W\uL_1 & \cdots & W\uL_N}_{(p+Ln)\by N\Tf},
\end{align}
where $\Tf = \Tbar - L+1$. %

From \eqref{eq:vector-form}, it holds that
\begin{align}\label{eq:Y-eqn}
\bY\uL = \bE\uL + G\uL \bU\uL + F\uL \bW\uL.
\end{align}

\section{Algorithm}
We propose and analyze \cref{alg:sysid} for learning models of clusters of linear systems. It consists of the following steps: 
\begin{enumerate}
\item Estimate the first $L_1$ Markov parameters for each observed system individually with a standard least squares estimator, which can be expressed in terms of the pseudoinverse of the Toeplitz input matrix $U_i^{(L_1)}$:
$\Ghat_i^{(L_1)} = Y^{(L_1)}_i\pseudo{(U_i^{(L_1)})}$, for system $i \in [\Ntot]$.
\item Cluster individual trajectories using the Euclidean distance of the first $L+1$ Markov parameters, $\Fnorm{\Ghat\u{L_1}_i - \Ghat\u{L_1}_j}$, for $i, j \in [\Ntot]$ as a measure of separation. A standard k-means clustering algorithm or variant \cite{lloyd1982least} suffices for this step. %

\item Compute a refined estimate of the first $L_2 + 1$ Markov parameters for each cluster via least squares, based on the multiple trajectories assigned to each cluster: 
$\Ghat^{(L_2)} = \bY^{(L_2)}\pseudo{(\bU^{(L_2)})}$.

\item Obtain a state space realization based on the refined Markov parameter estimates by running the Ho-Kalman algorithm (e.g., \cite[Algorithm 1]{oymak2021revisiting}) to obtain an order $n$ realization $(\Chat_{(k)}, \Ahat_{(k)}, \Bhat_{(k)}, \Dhat_{(k)})$ for each cluster model.
\end{enumerate}

The discussion in Section \ref{sec:interp} provides guidance on selecting parameters $L_1$ and $L_2$ in terms of $N$ and $\Tbar$, motivated by the finite sample guarantees of Theorem \ref{thm:MAIN-RESULT}. In general, the clustering step requires shorter trajectories than the refined estimation step.

\begin{algorithm}[h]
\caption{System Identification for Clustered Partially Observed Systems}\label{alg:sysid}
\begin{algorithmic}[1]
\Require Sampled trajectories $\set{u_i(0),(u_i(t), y_i(t))\mid t \in [\Tbar]}$, for $i \in [\Ntot]$; system order $n$; number of clusters $K$; number of Markov parameters $L_1\geq 1, L_2\geq 2n+1$; 
\Ensure Cluster assignments $(\hat\kappa_1, \hat\kappa_2, \ldots, \hat\kappa_{\Ntot}) \in [K]^{\Ntot}$; Markov parameters and balanced realizations $\set{\Ghat\u{L_2}, (\Chat_{(k)}, \Ahat_{(k)}, \Bhat_{(k)}, \Dhat_{(k)})}_{ k \in [K]}$ 

\For{$i \in [\Ntot]$}
	\State $\Ghat_i^{(L_1)} \gets Y^{(L_1)}_i\pseudo{(U_i^{(L_1)})}$ \Comment{Individual estimates}
\EndFor
\State $(\hat\kappa_1, \hat\kappa_2, ..., \hat\kappa_{\Ntot}) \gets$ k-means($\set{\Ghat_i^{(L)}: i \in [N]}, \Fnorm{\cdot}$) \Comment{Cluster assignments}
\For{$k \in [K]$}
	\State $\bY^{(L_2)} \gets \matb{Y^{(L_2)}_{\kappa_{i_1}} & \cdots & Y^{(L_2)}_{\kappa_{i_N}}}$ %
	\State $\bU^{(L_2)} \gets \matb{U^{(L_2)}_{\kappa_{i_1}} & \cdots & U^{(L_2)}_{\kappa_{i_N}}}$ 
	\Statex where $\kappa_{i_j} = k$ for $j\in [N]$.
	\State $\Ghat^{(L_2)} \gets \bY^{(L_2)}\pseudo{(\bU^{(L_2)})}$  \Comment{Refined estimates}
	\State $(\Chat_{(k)}, \Ahat_{(k)}, \Bhat_{(k)}, \Dhat_{(k)}) \gets $ Ho-Kalman$\lrp{\Ghat^{(L_2)}, n}$ %
\EndFor
\end{algorithmic}
\end{algorithm}

\section{Results}
\subsection{Estimation with Multiple Trajectories}
We provide end-to-end finite sample bounds for estimating Markov parameters and state space realizations for clusters of stable linear systems. These results rely on Theorem \ref{thm:MAIN-RESULT}, which bounds the error in estimating Markov parameters from multiple independent trajectories.

\begin{theorem}\label{thm:MAIN-RESULT}
	Suppose we have $N$ independent trajectories of length $\Tbar$ sampled from a model $M$ with realization $(C, A, B, D)$. Let $1 \leq L \leq \Tbar$, and define $\Tf \coloneq \Tbar - L + 1$. Then for any $\eps>0, \delta \in (0,1/e)$, when 
\begin{align}
	N\Tf &\gsim \frac{L}{\eps^2 \sigma_u^2} \lrp{p+(m+n)L+\ln\lrp{\frac{NT}{\delta}}}^2 \\
	&\quad \quad \cdot \lrp{\opnorm{CA^L}^2\opnorm{\Gamma_{\infty}}+ \sigma_w^2 \opnorm{\Gamma_{\infty}^O}}, 
\end{align}
the least squares estimator $\Ghat\uL= \bY\uL\pseudo{\bU\uL}$ of \eqref{eq:Y-eqn} satisfies
	\begin{align}
		\Pr{\big\lVert{\Ghat\uL - G\uL}\big\rVert_2 \lsim \eps} \geq 1-\delta.
	\end{align}
\end{theorem}
\vspace{1em}

The proof of Theorem \ref{thm:MAIN-RESULT} is found in Section \ref{sec:main-proof}. We follow a similar strategy to \cite{oymak2021revisiting}, but extend it to the setting of multiple trajectories. In a crucial step, we prove a restricted isometry property of the Toeplitz input matrix $U$ similar to \cite{oymak2021revisiting, sarkar2021finite}, but with a more streamlined approach, where we partition $U$ into $L+1$ embedded Page matrices and exploit concentration inequalities for Gaussian random matrices. 

\subsection{End-to-end Guarantees}
\begin{proposition}\label{prop:end2end}
	Suppose we have $\Ntot=NK$ trajectories of length $T$, with $N$ trajectories generated from each of the $K$ models $M_{(k)} \in \bigM$. Suppose the models in $\bigM$ are separated by at least $\Delta$ in their first $L_1+1$ Markov parameters as in \eqref{eq:cluster-distance}. Let $L_1, L_2$ be the number of Markov parameters estimated in Algorithm \ref{alg:sysid}, where $1 \leq L_1, L_2 \leq T$.

For every $\eps >0, \delta \in (0,1/e)$, if
\begin{align}\label{eq:first-T}
T-L_1 \gtrsim \frac{L_1 \Phi^{(L_1)} }{\Delta^2 \sigma_u^2} \lrp{p + L_1(m+n) + \ln\lrp{\frac{T}{\delta}}}^{2},
\end{align} 
and if the $k$-means algorithm converges to its global minimum, then \cref{alg:sysid} recovers cluster membership exactly with probability at least $1-\delta$.
Here, $\Phi^{(L)} \coloneq \max_{k \in [K]} \opnormBase{C_{(k)} A_{(k)}^L}^2 \opnormBase{\Gamma_\infty(M_{(k)})} + \sigma_w^2 \opnormBase{\Gamma^O_\infty(M_{(k)})}$.

\vspace{0.25em}
Furthermore, if
\begin{align}\label{eq:second-NT}
N(T-L_2) \gtrsim \frac{L_2 \Phi^{(L_2)} }{\eps^2\sigma_u^2} \lrp{p + L_2(m+n) + \ln\lrp{\frac{NT}{\delta}}}^{2},
\end{align}
there exists a permutation $\pi(\cdot)$ on $[K]$ such that the estimates of the first $L_2+1$ Markov parameters of the cluster models in $\bigM$ returned by \cref{alg:sysid} satisfy
\begin{align}
\Pr{\sup_{k \in [K]} \opnorm{\Ghat_{\pi(k)}\u{L_2} - G_k\u{L_2}} \lsim \eps} \gsim 1-\delta.
\end{align}
Additionally, when $L_2 \geq 2n+1$, 
\cref{alg:sysid} returns a realizations $\setBase{\hat C_{(k)}, \hat A_{(k)}, \hat B_{(k)}, \hat D_{(k)} \mid k \in [K]}$ such that for each $k$, for some invertible matrix $Q_k\in \R^{n \by n}$,
\begin{align}
& \max \Big\{\big\lVert{C_{(k)} - \hat \Cscr_{\pi(k)} Q_k}\big\rVert_\mathrm{F}, \big\lVert{A_{(k)} - \inv{Q_k}\hat \Ascr_{\pi(k)} Q_k}\big\rVert_\mathrm{F},\\
&\quad \big\lVert{B_{(k)} - \inv{Q_k} \hat \Bscr_{\pi(k)}}\big\rVert_\mathrm{F}, \big\lVert{D_{(k)}, - \hat \Dscr_{\pi(k)}}\big\rVert_\mathrm{F} \Big\} \\
& \quad \quad \lsim \sqrt{\frac{nL_2}{\sigma_n^2(H)}} \eps,
\end{align}
where $H_k$ is the $L_2/2 \by L_2/2$ principal submatrix of the Hankel matrix for system $M_{(k)}$.
\end{proposition}

The proof of Proposition \ref{prop:end2end}, provided in Section \ref{sec:proof:end2end}, proceeds by tracing parameter estimation errors through the four main steps of \cref{alg:sysid}. Theorem \ref{thm:MAIN-RESULT} provides bounds on the least squares estimator for individual system Markov parameters and the refined estimates Markov parameters within each cluster. %
The assumed separation between cluster models combined with the high probability estimation error bounds implies that the global minimizer of the sum of squared errors objective of $k$-means, based on the estimated Markov parameters, matches the true cluster assignments of the observed trajectories. While standard $k$-means algorithms are not guaranteed to find this global minimum, even when the clusters are well-separated, we assume a good initialization that leads to this minimum. In practice, this can be achieved with multiple initialization schemes.   %

\subsection{Interpretation of Results}\label{sec:interp}
\paragraph{Decoupling data requirements for clustering and estimation}
\cref{alg:sysid} allows us to decouple the data requirements for clustering from that for estimating a state space realization. %
In particular, Theorem \ref{thm:MAIN-RESULT} shows learning the first $L+1$ Markov parameters requires a trajectory length $T$ that scales at most as $O(L^3)$. It is likely that the cubic dependence on $L$ is an artifact of the proof, but $T$ must scale at least linearly with $L$. Since estimating an order $n$ realization requires $L_2 = 2n+1$ Markov parameters, $\Tbar$ must scale polynomially in $n$. Meanwhile, if the $K$ cluster models are well-separated in their first $L_1+1$ Markov parameters, where we only require $1 \leq L_1 \leq \Tbar/m$, then the number of samples needed for clustering is much less than for estimation.
If in practice, systems have varying trajectory lengths, we can set $L_1 = \min_{i \in [\Ntot]} \Tbar_i / m$ to accomodate all trajectory lengths, while restricting the realization step to sufficiently long trajectories within each cluster. 
Regarding the scaling of $T$ with $L$, it is a point of future work to find lower bounds to evaluate the tightness of our sample complexity results.

\paragraph{Trade-offs between $N$ and $T$ in estimation error guarantees}
Note that with a single length $T$ trajectory, Theorem \ref{thm:MAIN-RESULT} only guarantees that we can learn $L_1$ impulse response parameters up to accuracy $\eps_1$ where $L_1/\eps_1 \leq \sqrt{T}$. However, when we have observations from $N$ trajectories, we can learn $L_2$ parameters to accuracy $\eps_2$ where $L_2/\eps_2 \leq \sqrt{NT}$. In other words, we have a scaling of 
$$\frac{L_2}{\eps_2} \approx \sqrt{N} \frac{L_1}{\eps_1}$$
when $T$ is fixed.
This relation can be viewed as leveraging the $N$ trajectories to either increase the accuracy of estimating a fixed length of impulse response, or of increasing the length of impulse response that can be estimated with a fixed accuracy, for a given trajectory length $T$. 

\section{Proof of Theorem \ref{thm:MAIN-RESULT}}\label{sec:main-proof}
In this section, we drop the superscript $(L)$ in all relevant variables, %
where $L+1$ is the number of Markov parameters to be estimated.
When $U$ has full row rank, $UU'$ is invertible, and $\pU = U'\inv{(UU')}$. This holds almost surely when $U$ is a block Toeplitz matrix with random Gaussian entries \cite{pan2015estimating}, and the singular values of $UU'$ can be lower bounded by a positive number with high probability. 
Then the operator norm of the least squares estimation error of the matrix of the first $L+1$ Markov parameters can be upper bounded as $\opnormBase{\Ghat\uL - G\uL} \leq \opnormBase{EU'}\opnormBase{\inv{(UU')}} + \opnormBase{F\uL} \opnormBase{WU'} \opnormBase{\inv{(UU')}}$ by expanding the expression for $\bY\u{L_2}$ using \eqref{eq:Y-eqn}.

Following a proof strategy similar to that of \cite[Theorem 3.2]{oymak2021revisiting}, we upper bound each of the factors in the error decomposition separately, as stated in Propositions \ref{prop:UU-bound}, \ref{prop:WU-bound}, and \ref{prop:EU-bound} below. %
Combining these results, and using that $\opnormBase{F\uL} \leq \opnormBase{\Gamma_{\infty}^O}^{1/2}$, we have that with probability at least $1-\delta$, when $\delta \in (0, 1/e)$, 
\begin{align}
	\opnorm{\Ghat\uL - G\uL} &  \lsim \frac{\sqrt{L} \lrp{p+(m+n)L + \ln\lrp{NT/\delta}}}{\sigma_u \sqrt{N\Tf}}\\
	&\quad \cdot \lrp{\opnorm{CA^L}\opnorm{\Gamma_\infty}^{1/2} + \sigma_w \opnorm{\Gamma_{\infty}^O}^{1/2}}.
\end{align}
Expressing this result in terms of a sample complexity in $N\Tf$ for a given error $\eps>0$ gives the claim.

\paragraph{Conditioning of input covariance matrix}%
We first prove a high probability bound on the minimum singular value of the matrix $UU'$, showing that it is well-conditioned. Recall that for all $i \in [N], t \in [0{:}\Tbar]$, $u_{i,t}$ is independently distributed as $\Nscr(0,\sigma_u^2I_m)$. Ideally, we would like to use concentration bounds for the singular values of Gaussian random matrices, but the Toeplitz structure of each block $U_i$ creates statistical dependencies across elements of the matrix. We get around this complication by partitioning $U$ into $L+1$ embedded Page matrices, such that the elements within each submatrix are statistically independent, and then leveraging bounds on the singular values of Gaussian random matrices. 
\begin{proposition}\label{prop:UU-bound}
For any $\alpha \in (0,1)$, $\delta \in (0,1)$, if 
\begin{align}
N\Tf \gsim \frac{L}{\alpha^2}\lrp{mL + \ln\lrp{\frac{L}{\delta}}},
\end{align}
then with probability at least $1-\delta$,
\begin{align}
\opnorm{(UU')^{-1}} \leq \frac{1}{\sigma_{\min}(UU')} \leq \frac{1}{(1-\alpha)^2\sigma_u^2 N\Tf}.
\end{align}
\end{proposition}
\vspace{0.5em}

\paragraph{Effect of noise}
The quantity $\opnorm{WU'}$ can be bounded with high probability by using the same partitioning of $W$ and $U$ such that the submatrix of each partition is composed of zero-mean independent subgaussian entries, for which we can apply subgaussian concentration inequalities. 
\begin{proposition}\label{prop:WU-bound}
For any $\delta \in (0,1)$, with probability at least $1-\delta$, when $\Tf N/L \gtrsim \lrp{p+L(n+m)} \ln(L/\delta)$,
\begin{align}
\opnorm{WU'} \lsim \sigma_w\sigma_u \sqrt{N\Tf L(p+L(n+m) + \ln(L/\delta))}.
\end{align}
\end{proposition}
\vspace{0.5em}

\paragraph{Effect of truncated impulse response}
The $EU'$ term in the estimation error features $CA^L x_{i,t-L}$, which arises from the fact that the length $L$ truncated impulse response does not account for the effect of the unobserved states $x_{i,t-L}$, for $i \in [N], t \in [L:\Tbar]$. Our approach in bounding $\opnorm{EU'}$ is similar to that of \cite[Theorem 4.2]{oymak2021revisiting}, but with a different presentation of the results.

\begin{proposition}\label{prop:EU-bound}
For $\delta \in (0,1/e)$, with probability at least $1 - \delta$, 
\begin{align}
	\opnorm{EU'} & \lsim \sigma_u \sqrt{NTL} \opnorm{CA^L} \opnorm{\Gamma_{\infty}}^{1/2}\\
	&\quad\quad \cdot \lrp{p+mL+\ln(NT/\delta)}.
\end{align}
\end{proposition}
\vspace{0.5em}
The proofs of Propositions \ref{prop:UU-bound}, \ref{prop:WU-bound} and \ref{prop:EU-bound} can be found in Sections \ref{sec:proof:UU-bound}, \ref{sec:proof:WU-bound} and \ref{sec:proof:EU-bound}, respectively. 

\section{Simulations}
We evaluate the performance of \cref{alg:sysid} in estimating Markov parameters of clusters of linear systems 
through simulations. 
$K$ single-input single-output state space models of order $n=3$ were randomly generated, and simulated with process and observation noise to generate $N$ unlabeled observations of length $T$ input-output trajectories from each of the $K$ different models. The $A$ matrices were randomly generated and scaled to have their spectral radii between $0.6$ and $0.9$, while $B$ and $C$ matrices were scaled to have unit spectral norm. Inputs were drawn as i.i.d.\ isotropic gaussians with $\sigma_u = 1$, while process and measurement noise levels were $\sigma_{w\u1} = 0.15, \sigma_{w\u2} = 0.2$, respectively. 

Figure \ref{fig:err_versus_T_and_cluster_width} shows the average model estimation error $\sum_{k=1}^K \lVert \ghat_k^{(L)} - g_k^{(L)}\rVert_2/K$ for the first $L=4$ Markov parameters of the $K=3$ cluster models, plotted against $T$ for different values of $N$, and for varying cluster widths, i.e., distances between systems in the same cluster. 
The estimation error decreases as both the trajectory length $T$ and the number of trajectories $N$ increase, which is also seen in the plot of level sets of the estimation error in Figure \ref{fig:err_level_set}. This indicates that even with short trajectories, increasing $N$ can improve cluster estimation. The algorithm is robust to moderate cluster widths, though as expected, the performance degrades as the cluster widths approach the cluster separations.
\begin{figure}[h]
	\centering
	\includegraphics[width=\columnwidth]{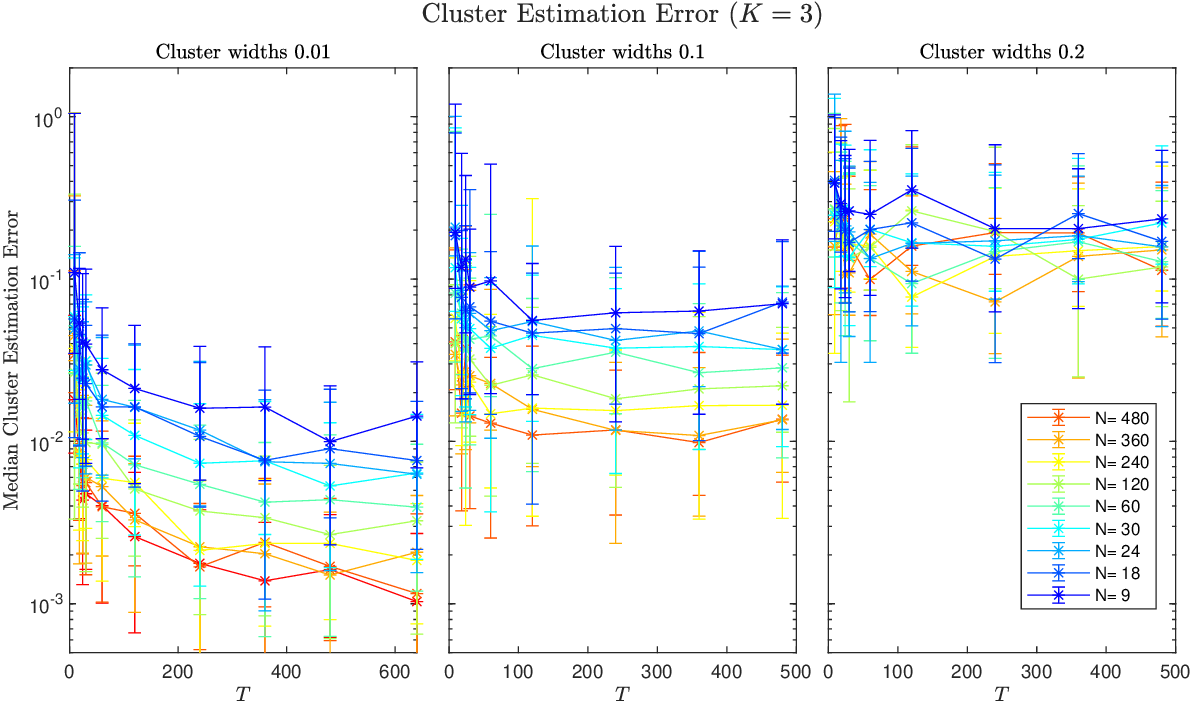}
	\caption{Model estimation error vs. trajectory length $T$, for various values of $N$. From the left to right panel, the cluster widths increase. 20 trials per point. 25th to 75th percentile error bars are plotted.}
	\label{fig:err_versus_T_and_cluster_width}
\end{figure}

\begin{figure}[hb]
	\centering
	\includegraphics[width=1\columnwidth]{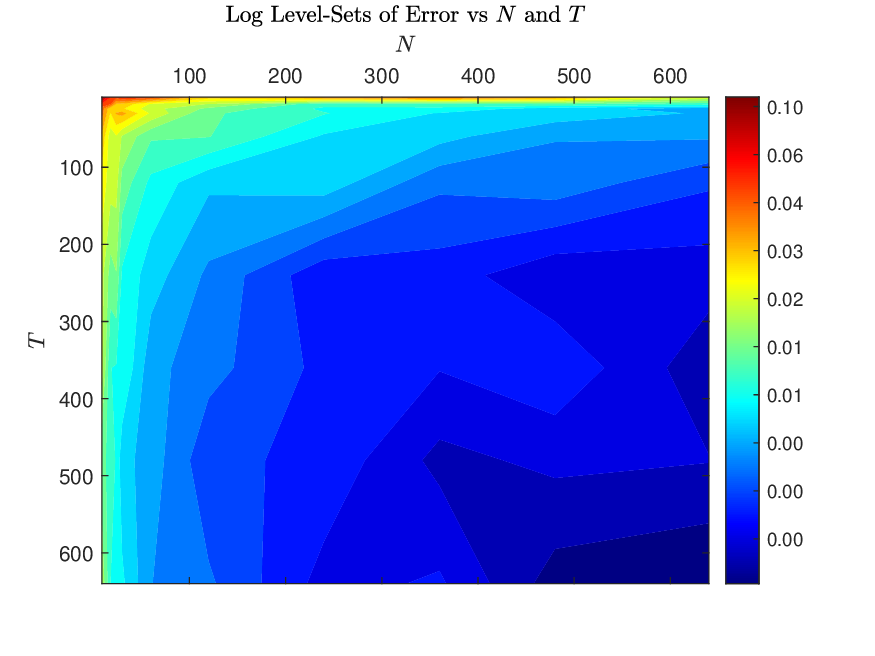}
	\vspace{-0.8cm}
	\caption{Level sets of model estimation error in $N$ and $T$, for the same setting in Fig. \ref{fig:err_versus_T_and_cluster_width}. Error decreases in both $N$ and $T$.}
	\label{fig:err_level_set}
\end{figure}

Finally, we compare the performance of Algorithm \ref{alg:sysid} with a moment-based approach for estimating mixtures of linear systems. The algorithm, proposed in \cite{rui2024finite}, extends spectral methods for mixtures of linear regression and involves decomposing a third-order tensor constructed the data.
Figure \ref{fig:comparison} plots the average model estimation error for both methods in estimating the first $L=7$ Markov parameters of $K=3$ models from $N$ trajectories of length $T$, where $N$ and $T$ vary.  While the moment-based approach comes with theoretical global convergence guarantees, in contrast to the local clustering-based analysis in this paper, Algorithm \ref{alg:sysid} is empirically more sample efficient and accurate.

\begin{figure}[hbtp]
	\centering
	\includegraphics[width=\columnwidth]{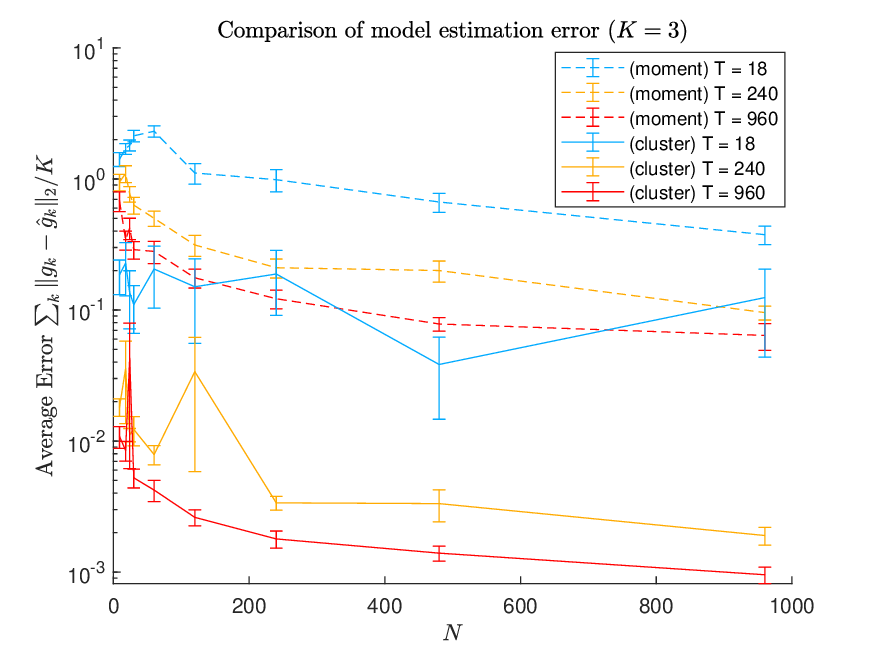}
	\caption{Comparison of model estimation error of Algorithm \ref{alg:sysid} (``cluster'') with that of a moment-based tensor decomposition estimator (``moment'') for mixtures of linear systems \cite{rui2024finite}, for various trajectory lengths $T$, plotted against the number $N$ of observed trajectories. Standard errors for 15 trials are shown.}
	\label{fig:comparison}
\end{figure}

\bibliographystyle{IEEEtran}
\bibliography{regression_problem_bib}

\onecolumn
\section{Appendix} 

\subsection{Ho-Kalman Realization}
We may recover a balanced minimal realization $(C, A, B, D)$ for an order $n$ system given its Markov parameters $\set{D, CB, CAB, ...}$ by using the classic Ho-Kalman realization algorithm \cite{ho1966effective}, which factors the following Hankel matrix formed by the Markov parameters:
\begin{align}
\Hscr\u{L,L} = \matb{CB & CAB & \cdots & CA^{L-1}B \\
CAB & CA^2B & \cdots & CA^{L}B \\ 
\vdots & & \ddots & \\
CA^{L-1}B & CA^{L}B & \cdots & CA^{2L-2}B}.
\end{align}
When $L$ is large enough (i.e., $L \geq n$), $\Hscr\u{L,L}$ will be rank $n$, and we can obtain the observability and controllability matrices by factorization of $\Hscr\u{L,L}$ via singular value decomposition, and in turn recover the matrices $C$ and $B$, and finally solve for $A$. 

To estimate an order $n$ realization from set of the first $2L+2$ estimated Markov parameters in $\Ghat\u{2L+1} =\matb{\Ghat_0 & \Ghat_1 \cdots & \Ghat_{2L} & \Ghat_{2L+1}}$, note that $\Dhat$ can be read off as $\Ghat_0$, and for the remaining terms, we similarly form the Hankel matrix of Markov parameters: 
\begin{align}
\hat{\Hscr}\u{L,L}  = \matb{\Ghat_0 & \Ghat_1 & \cdots & \Ghat_L \\
\Ghat_1 & \Ghat_2 & \cdots & \Ghat_{L+1} \\ 
\vdots & & \ddots & \\
\Ghat_L & \Ghat_{L+1} & \cdots & \Ghat_{2L}}.
\end{align}
Using the Ho-Kalman algorithm given in \cite[Algorithm 1]{oymak2021revisiting}, we can recover an estimated realization $(\Chat, \Bhat, \Ahat)$ when $L \geq n$. \cite{oymak2021revisiting} analyzes the robustness of the Ho-Kalman algorithm with respect to the estimation error of the first $2L+2$ Markov parameters. Their result, tailored to our setting, is stated in Proposition \ref{prop:ho-kalman}.

\begin{proposition}[{\cite[Theorem 5.2]{oymak2021revisiting}}]\label{prop:ho-kalman}
Given a system of order $n$, let $L\geq n$ and suppose we have estimates of the first $(2L+2)$ Markov parameters of a system, concatenated in the matrix $\Ghat\u{2L+1}$. 
Define $\Hscr^{-}$ and $\hat{\Hscr}^{-}$ to be the submatrices consisting of the first $L$ block columns of $\Hscr\u{L,L}$ and $\hat{\Hscr}\u{L,L}$, respectively.
Let $G\u{2L+1}$ be the matrix of the true first $(2L+2)$ Markov parameters, 
and suppose that $$\opnorm{G\u{2L+1} - \Ghat\u{2L+1}} \leq \sqrt{n}\sigma_n(\Hscr^{-})$$
Then the Ho-Kalman algorithm \cite[Algorithm 1]{oymak2021revisiting} on $\Ghat\u{2L+1}$ yields an order $n$ realization $(\Chat, \Ahat, \Bhat, \Dhat)$ such that there exists an invertible matrix $Q \in \R^{n \by n}$ such that 
\begin{align}
	\Fnorm{A - \inv{Q}\Ahat Q} &\lsim \frac{ \sqrt{nL}\opnorm{\Hscr\u{L,L}}}{\sigma_n^2(\Hscr^{-})} \opnorm{G\u{2L+1} - \Ghat\u{2L+1}} \\
	\max \set{\Fnorm{C - \Chat Q}, \Fnorm{B - \inv{Q} \Bhat}} &\lsim \frac{\sqrt{nL}\opnorm{G\u{2L+1} - \Ghat\u{2L+1}}}{\sqrt{\sigma_n(\Hscr^{-})}} \\
	\Fnorm{D - \Dhat} & \lsim \sqrt{n} \opnorm{G\u{2L+1} - \Ghat\u{2L+1}}.
\end{align}
\end{proposition}

\subsection{Bounding \texorpdfstring{$UU'$}{UU'}}\label{sec:proof:UU-bound}
\begin{proof}[Proof of Proposition \ref{prop:UU-bound}]
Note that $U$ is a concatenation of $N$ block-Toeplitz matrices $U_i$, for $i \in [N]$:
\begin{align}\label{eq:toeplitz}
U_i = \lrb{\ubar_{i,L}\ \cdots \ \ubar_{i,\Tbar}} = \matb{u_{i,L} & u_{i,L+1} & \cdots & u_{i,\Tbar}  \\
u_{i,L-1} & u_{i,L} & \cdots & u_{i,\Tbar-1} \\
\vdots & \ddots & & \vdots \\
u_{i,0} & u_{i,1} & \cdots & u_{i,\Tbar-L}},
\end{align}
which is a $(L+1)m \by \Tf$ matrix, 
recalling that $\Tf = \Tbar-L+1$ is the ``effective'' number of samples per trajectory to estimate the first $L+1$ Markov parameters.

We partition $U$ into $L+1$ embedded Page matrices. 
For each $k \in [0{:}L]$,
define the index set $\Jscr_k = \set{k + j(L+1) -1 \mid j = [d]}$, where $d = \floor{\Tf/(L+1)}$. For simplicity, let us assume that $\Tf$ is divisible by $(L+1)$ since we can otherwise truncate $\Tf$ to the preceding closest multiple of $(L+1)$. Then for example, for system $i$ and for the $k$th index set, we have the submatrix indexed by $\Jscr_k$,
\begin{align}\label{eq:page-partition}
U_{i,\Jscr_k} = \matb{\ubar_{i, L+k} & \ubar_{i,2(L+1)-1+k} & \cdots & \ubar_{i,d(L+1)-1+k}},
\end{align}
where the $j$th column is $\ubar_{j(L+1)-1+k}$.
Let $U_{\Jscr_k} \coloneq \lrb{U_{1,\Jscr_k} \cdots U_{N,\Jscr_k}}$ be the matrix formed by concatenating matrices $U_{i,\Jscr_k}$ for $i \in [N]$.
We can then bound the singular values of $U$ in terms of the singular values of the constituent matrices $\tilde{U}_k$, $k \in [0{:}L]$, and in particular, by Lemma \ref{lem:decompose},
\begin{align}\label{eq:pf-sum}
\sigma_{\min}^2(U) \geq \sum_{k = 0}^L \sigma_{\min}^2 \lrp{\tilde{U}_{i, k}}.
\end{align}
 
For each $k \in [0{:}L]$, $U_{\Jscr_k}$ is a $m(L+1)\by dN$ Page matrix and has independent, identically distributed elements with distribution $\Nscr(0, \sigma_u^2)$.
Applying known bounds on the singular values of Gaussian random matrices $\tilde{U}_k$ (Lemma \ref{lem:gaussian-matrix-sv}) %
with a union bound over $k \in [0{:}L]$, we get that when $dN\geq m(L+1)$, it holds that
\begin{align}
1-\frac{\sigma_{\min}(U_{\Jscr_k})}{\sigma_u \sqrt{dN}} \leq \sqrt{\frac{m(L+1)}{dN}} + \sqrt{\frac{2\ln(2(L+1)/\delta)}{dN}}
\end{align}
for all $k \in [0{:}L]$ 
with probability at least $1-\delta$.
We upper bound the right hand side of this inequality by $\alpha$, and note that $\sigma_{\min}(U_{\Jscr_k}) \geq (1-\alpha) \sigma_u \sqrt{dN}$ with probability at least $1-\delta$ when 
\begin{align}
dN \geq \frac{2}{\alpha^2} \lrp{m(L+1) + 2\ln(2(L+1)/\delta)}.
\end{align}
Plugging these bounds into \eqref{eq:pf-sum} and recalling $d = \Tf/(L+1)$ completes the proof.
\end{proof}

\begin{lemma}\label{lem:decompose}
Let $U$ be an $m\by n$ matrix, with $n > m$, and let $(\Iscr_1, ..., \Iscr_q)$ be a partition of the index set $\set{1, 2, ..., n}$. If $A_k$ is the $m \by \abs{\Iscr_k}$ matrix whose columns correspond to the columns of $U$ indexed by $\Iscr_k$, then 
\begin{align}
\sqrt{\sum_{i=1}^q \sigma_{\min}^2 (A_i)} \leq \sigma_{\min}(U) \leq \sigma_{\max}(U) \leq \sqrt{\sum_{i=1}^q \sigma_{\max}^2 (A_i)}.
\end{align} 
\end{lemma}
\begin{proof}
Without loss of generality, we assume that $U$ can be written as
\begin{align}
U = \matb{A_1 & \cdots & A_q},
\end{align}
since singular values are preserved under the elementary matrix operations of permuting columns. 
Next, for $i \in [m]$, since singular values are preserved under matrix transpose, let us consider 
\begin{align}
\sigma_{\min}(U) = \inf_{v \in \Sscr^{m-1}} \enorm{v'U},\\
\sigma_{\max}(U) = \sup_{v \in \Sscr^{m-1}} \enorm{v'U}.
\end{align}
Since $v'U = \matb{v'A_1 & \cdots & v'A_q}$,
we can decompose $\enorm{v'U}^2 = \sum_{k=1}^q \enorm{v'A_k}^2$. But for each $k$, $\sigma_{\min}(A_k) \leq \enorm{v'A_k} \leq \sigma_{\max}(A_k)$. This implies that 
\begin{align}
\sum_{k=1}^q \sigma_{\inf}^2(A_k) \leq \min_{v \in \Sscr^{m-1}} \sum_{k=1}^q \enorm{A_k v}^2 = \sigma_{\min}^2 (U) \leq \sigma_{\max}^2(U) =\sup_{v \in \Sscr^{m-1}} \sum_{k=1}^q \enorm{A_k v}^2 \leq \sum_{k=1}^q \sigma_{\max}^2(U_k). 
\end{align}
Taking the square root of each term in the chain of inequalities gives the claim. 
\end{proof}

\begin{lemma}[Theorem 9.26 in \cite{foucart2013mathematical}]\label{lem:gaussian-matrix-sv}
Let $A$ be an $m \by n $ matrix with iid standard normal entries, with $n > m$. Then for all $t>0$,
\begin{align}
\Pr{\frac{\sigma_{\max}(A)}{\sqrt{n}} \geq 1 + \sqrt{m/n} + t} &\leq e^{-nt^2/2}, \\
\Pr{\frac{\sigma_{\min}(A)}{\sqrt{n}} \leq 1 - \sqrt{m/n} - t} &\leq e^{-nt^2/2}.
\end{align}
\end{lemma}

Proposition \ref{prop:UU-bound} provides a lower bound on the singular values of $UU'$. For completeness, in the next lemma we also provide an upper bound for $UU'$. 
\begin{lemma}
For every $\alpha, \delta \in (0,1)$, 
when $$N\Tf \gtrsim \frac{L}{\alpha^2} \lrp{mL + \ln\lrp{\frac{L}{\delta}}},$$
it holds with probability at least $1-\delta$ that
\begin{align}
(1-\alpha)^2 \sigma_u^2 N\Tf I_{m(L+1)} \prec UU' \preceq (1+\alpha)^2 \sigma_u^2 N\Tf I_{m(L+1)}.
\end{align}
\end{lemma}
\begin{proof}
For ease of notation let $d = \Tf/(L+1)$. 
Similar to the proof of Proposition \ref{prop:UU-bound}, since for each $m(L+1) \by Nd$ submatrix $U_{\Jscr_j}$ in the partition of the columns of $U$  consists of iid zero-mean $\sigma_u^2$-subgaussian random variables, Lemma \ref{lem:gaussian-matrix-sv} implies that when $Nd > m(L+1)$, for every $j \in [L]$, for every $\delta \in (0,1)$, with probability at least $1-\delta$, 
\begin{align}
\frac{\sigma_{\max}(U_{\Jscr_j})}{\sigma_u \sqrt{Nd}} - 1 \leq \sqrt{\frac{m(L+1)}{dN}} + \sqrt{\frac{2}{d} \ln \lrp{\frac{2}{\delta}}}\xpln{and}\\
\1-\frac{\sigma_{\min}(U_{\Jscr_j})}{\sigma_u \sqrt{dN}} \leq \sqrt{\frac{m(L+1)}{dN}} + \sqrt{\frac{2}{dN}\ln\lrp{\frac{2}{\delta}}}.
\end{align}
Let $\alpha \in (0,1)$. Taking a union bound over all partition elements $j \in [L]$, we have that when 
\begin{align}
dN > \frac{2}{\alpha^2}\lrp{m(L+1) + \frac{2}{dN}{\ln\lrp{\frac{2L}{\delta}}}}, 
\end{align}
with probability at least $1-\delta$, for all $j \in [L]$, 
\begin{align}
(1-\alpha) \sigma_u \sqrt{dN} \leq \sigma_{\min}(U_{\Jscr_j}) \leq \sigma_{\max}(U_{\Jscr_j}) \leq \sigma_u \sqrt{dN}(1+\alpha).
\end{align}
Combining the bounds  for all $U_{\Jscr_j}$ with Lemma \ref{lem:gaussian-matrix-sv} gives us the claimed result. 
\end{proof}

\subsection{Bounding \texorpdfstring{$WU'$}{WU'}}\label{sec:proof:WU-bound}

\begin{proof}[Proof of Proposition \ref{prop:WU-bound}]
Splitting the terms of $WU'$ based on the partition $\set{\Jscr_j \mid j \in [L]}$ of the columns of $U$ and $W$ and using the triangle inequality, we have $\opnorm{WU'} \leq \sum_{j =1}^L \opnorm{W_{\Jscr_j} U_{\Jscr_j}'}$. Each summand is a product of two independent $(p+Ln)\by dN$ and $m(L+1)\by dN$ matrices each with iid zero-mean subgaussian entries, where we set $d \coloneq \Tf/(L+1)$. 

Then by Lemma \ref{lem:rando-concentration} and a union bound over $j \in [L]$, for any $\delta > 0$, with probability at least $1-\delta$, 
\begin{align}
\opnorm{WU'}  \lsim  \sigma_w \sigma_u  NdL \max \lrp{\sqrt{\frac{\ln( 9^{p + m + L(n+m)}L/\delta)}{Nd}}, \frac{\ln( 9^{p + m + L(n+m)}L/\delta)}{Nd}}.
\end{align}
If $Nd \geq \ln( 9^{p + m + L(n+m)}L/\delta)$ then the result simplifies to 
\begin{align}
\Pr{\opnorm{WU'}  \lsim  \sigma_w \sigma_u  NdL \sqrt{\frac{\ln( 9^{p + m + L(n+m)}L/\delta)}{Nd}}} \geq 1-\delta.
\end{align}
Recalling that $Nd = N\Tf/(L+1)$ gives the result. 
\end{proof}

\begin{lemma}\label{lem:rando-concentration}
Let $A_{N \by M_1}$ and $B_{N \by M_2}$ be independent random matrices with iid zero-mean subgaussian entries with variance proxies $\sigma_a^2$ and $\sigma_b^2$, respectively. 
For every $\delta > 0$, with probability at least $1-\delta$, 
\begin{align}
\opnorm{A'B} \lsim \sigma_a \sigma_b N \max \lrp{\sqrt{\frac{\ln( 9^{M_1 + M_2}/\delta)}{N}}, \frac{\ln( 9^{M_1 + M_2}/\delta)}{N}}.
\end{align}
\end{lemma}

\begin{proof}
By \cite[Corollary 4.2.13]{vershynin2018high}, there exists a $1/4$-covering in the Euclidean norm, $\Cscr_1$, of the unit sphere $\Sscr^{M_1-1}$, and likewise a $1/4$-covering $\Cscr_2$ of $\Sscr^{M_2 -1}$, such that $|\Cscr_i| \leq 9^{M_i}$ for $i \in \set{1,2}$. 

We proceed by bounding $\brk{Aa, Bb}$ for all $(a,b) \in \Cscr_1 \by \Cscr_2$, and then transfer this bound to one on the operator norm using the following covering lemma \cite[Exercise 4.4.3]{vershynin2018high}:
\begin{align}\label{eq:opnorm-cover}
\opnorm{A'B} \leq 2 \sup_{a \in \Cscr_1, b \in \Cscr_2} \brk{Aa, Bb}.
\end{align}

For any $(a, b) \in \Cscr_1 \by \Cscr_2$, $Aa$ and $Bb$ are subgaussian random vectors with independent entries, and $\brk{Aa, Bb}$ is the sum of $N$ independent zero-mean subexponential random variables \cite[Lemma 2.7.7]{vershynin2018high} with parameter at most $8/3 \sigma_a \sigma_b$. Using Bernstein's inequality \cite[Theorem 2.8.1]{vershynin2018high} with a union bound over the covering, we have that for all $t \geq 0$, 
\begin{align}
\Pr{\sup_{(a,b) \in \Cscr_1 \by \Cscr_2} \brk{Aa, Bb} \geq \sigma_a \sigma_b N t} \leq 2 \abs{\Cscr_1}\abs{\Cscr_2} \exp \lrb{-c \min \lrp{\frac{t^2}{\sigma_a^2 \sigma_b^2}, \frac{t}{\sigma_a \sigma_b}} N}.
\end{align}
Then using \eqref{eq:opnorm-cover} and $\abs{\Cscr_i} \leq 9^{M_i}$ for $i = 1, 2$, we transfer this bound to the operator norm: for every $t>0$, 
\begin{align}
\Pr{\opnorm{A'B} \geq 2\sigma_a \sigma_b N t} \leq  2\cdot 9^{M_1 + M_2} \exp \lrp{-c \min\lrp{t^2, t} N}.
\end{align}
Letting $\delta = 2\cdot 9^{M_1 + M_2} \exp(-c\min(t^2,t)N)$, or equivalently, $t = c\max \lrp{\sqrt{\ln(2\cdot 9^{M_1 + M_2}/\delta)/N}, \ln(2\cdot 9^{M_1 + M_2}/\delta)/N}$, then we have
\begin{align}
\Pr{\opnorm{A'B} \gsim \sigma_a \sigma_b N \max \lrp{\sqrt{\ln( 9^{M_1 + M_2}/\delta)/N}, \ln( 9^{M_1 + M_2}/\delta)/N} } \leq \delta.
\end{align}
\end{proof}

\begin{lemma}\label{lem:innerprod}
Let $A$ and $B$ be $N\by M_1$ and $N \by M_2$ Gaussian random matrices with iid standard normal entries. For any $a \in \Sscr^{M_1-1}, b \in \Sscr^{M_2 -1}$, and for any $t>0$, 
\begin{align}
\Pr{\brk{Aa, Bb} \geq N t} \leq 2\exp \lrb{-c\min \lrp{t^2,t}N} 
\end{align}
where $c>0$ is an absolute constant.
\end{lemma}

\begin{proof}
First, note that for any fixed $a \in \Sscr^{M_1-1}, b \in \Sscr^{M_2 -1}$, $Aa$ and $Bb$ are independent $N$-dimensional vectors whose entries are independent standard normal random variables. 
Since the product of two sub-gaussian variables (c.f. \cite[Definition 2.5.6]{vershynin2018high}) is sub-exponential \cite[Lemma 2.7.7]{vershynin2018high}, the inner product $\brk{Aa, Bb}$ is a sum of $N$ sub-exponential RVs each with parameter at most 8/3. 
Using Bernstein's inequality \cite[Theorem 2.8.1]{vershynin2018high} yields the result. 
\end{proof}

\subsection{Bounding \texorpdfstring{$EU'$}{EU'}}\label{sec:proof:EU-bound}
\begin{proof}[Proof of Proposition \ref{prop:EU-bound}]
As in the proofs of Propositions \ref{prop:UU-bound} and \ref{prop:WU-bound}, we partition the $N\Tf$ columns of $U$ by the $L+1$ index sets $\set{\Jscr_j \mid j \in [0{:}L]}$ such that each submatrix $U_{\Jscr_j}$ consisting of the columns of $U$ indexed by $\Jscr_j$ is a $m(L+1)\by Nd$ Page matrix with iid $\Nscr(0, \sigma_u^2)$ entries, where $d = \floor{\Tf/(L+1)}$. 
Recall $\Tf = \Tbar - L + 1$ and for simplicity assume that $\Tf$ is divisible by $L+1$. Partitioning $E$ in the same way with each $E_{\Jscr_j}$ a $p \by Nd$ matrix, we can decompose
\begin{align}\label{eq:decomposeEU}
EU' = \sum_{j=0}^{L}E_{\Jscr_j} U_{\Jscr_j}'
\end{align}
and by the triangle inequality, it suffices to upper bound $\big\lVert{E_{\Jscr_j} U_{\Jscr_j}'}\big\rVert_2$ for each $j \in [0{:}L]$ in order to upper bound $\opnorm{EU'}$.

Fix a $j \in [0{:}L]$. We want to bound $\big\lVert{E_{\Jscr_j} U_{\Jscr_j}'}\big\rVert_2$. Let $\Cscr_1, \Cscr_2$ be 1/4-covers of $\Sscr^{p-1}, \Sscr^{m(L+1)}$, respectively, such that $\abs{\Cscr_1} \leq 9^p$ and $\abs{\Cscr_2} \leq 9^{m(L+1)}$ and 
\begin{align}
\opnorm{E_{\Jscr_j} U_{\Jscr_j}'} \leq 2\sup_{a \in \Cscr_1, b \in \Cscr_2} a'E_{\Jscr_j}U_{\Jscr_j}'b.
\end{align}
Define a mapping of pairs of indices indicating the partition $j\in [0{:}L]$ and element $k\in [d]$ of the partition to linear indices $\tau: (j,k) \mapsto k(L+1)+(j-1)\in [T]$. 
For a fixed $(a, b) \in \Cscr_1 \by \Cscr_2$, define $Z_{i,k} = \brk{a, CA^L x_{i,\tau(j,k-1)}}$ and $W_{i,k} = \brk{\ubar_{i, \tau(j,k)}, b}$ so that
\begin{align}
a'E_{\Jscr_j}U_{\Jscr_j}'b = \sum_{i=1}^N \sum_{k=1}^d Z_{i,k} W_{i,k}.
\end{align}
From Lemma \ref{lem:CAx-bound}, for a given $i \in [N], k \in [d]$ and $\delta \in (0, 1/e)$, with probability at least $1-\delta$, 
\begin{align}\label{eq:eu-step1}
	Z_{i,k}^2 \leq \enorm{CA^L x_{i, \tau(j, k-1)}}^2 \lsim \opnorm{CA^L}^2\opnorm{\Gamma_{\infty}} \ln(1/\delta) = \beta,
\end{align}
where we let $\beta \coloneq \opnorm{CA^L}^2\opnorm{\Gamma_{\infty}} \ln(1/\delta)$. Let us also define
$\alpha \coloneq \sqrt{2Nd\beta \sigma_u^2\ln(1/(Nd\delta))}$. %

Define the filtration with sigma algebras $\Fscr_{i,k}$ generated by $\set{x_{i,t}, u_{i,t} \mid t \leq \tau(j,k)}$. Note that $Z_{i,k}$ is adapted to $\Fscr_{i,k-1}$ and $W_{ik}$ is adapted to $\Fscr_{i,k}$, and $W_{i,k} \mid \Fscr_{i,k-1}$ is $\sigma_u^2$-subgaussian. 
By a martingale concentration lemma \cite[Lemma 4.4]{oymak2021revisiting} plus a union bound of \eqref{eq:eu-step1} over $i \in [N], k \in [d]$, we have that
\begin{align}
\Pr{\sum_{i,k} Z_{ik}W_{ik} \geq \alpha} &\leq \Pr{\sum_{i,k} Z_{i,k}^2 \leq Nd\beta} + \exp\lrp{-\frac{\alpha^2}{2\sigma_u^2 N d \beta}} \\
&\leq 2Nd\delta.
\end{align}
Plugging in $\alpha$ and $\beta$, we have that 
\begin{align}
	\Pr{\sum_{i,k} Z_{i,k} W_{i,k} \geq \sqrt{2Nd \sigma_u^2 \opnorm{CA^L}^2 \opnorm{\Gamma_{\infty}} \ln\lrp{\frac{1}{\delta}}\ln\lrp{\frac{1}{Nd\delta}}}} \leq 2Nd\delta.
\end{align}
Taking a union bound over all $(a,b) \in \Cscr_1 \by \Cscr_2$ and over all partitions $j \in [0{:}L]$, and plugging in $d = \Tf/(L+1)\leq T/L$, 
we have that with probability at least $1-\delta$, with $\delta < 1/e$,
\begin{align}
\opnorm{EU} &\lsim L \sqrt{Nd  \opnorm{CA^L}^2 \opnorm{\Gamma_{\infty}} \sigma_u^2 \ln\lrp{\frac{NdL9^{p+mL}}{\delta}} \ln \lrp{\frac{L9^{p+mL}}{\delta}}} \\
&\lsim \sqrt{NTL\opnorm{CA^L}^2\opnorm{\Gamma_{\infty}} \sigma_u^2 \lrp{p+mL+\ln\lrp{\frac{NT}{\delta}}}^2}.
\end{align}
\end{proof}

\begin{lemma}\label{lem:CAx-bound}
For a given $i \in [N], t \in [T]$, for any $\delta \in (0,1/e)$, 
\begin{align}
	\Pr{\enorm{CA^Lx_{i,t}}^2 \gsim \opnorm{CA^L}^2 \opnorm{\Gamma_{\infty}} \ln(1/\delta)} \leq 2\delta.
\end{align}
\end{lemma}
\begin{proof}
For every $i \in [N], t \in [T]$, we can write the state at time $t$ (under a fixed state space representation) as a sum of independent zero-mean subgaussian random variables:
\begin{align}
x_{i,t} &= \sum_{s=1}^t A^{s-1}\lrp{Bu_{i,t-s} + w\u{1}_{i,t-s}}, \\
&= H_t z_{i,t},
\end{align}
where we define 
\begin{align}
H_t &= \lrb{\sigma_u[B,\ AB,\ \cdots,\ A^{t-1}B],\ \sigma_w[I_n,\ A,\ \cdots,\ A^{t-1}]}\\
z_{it}' &= \lrb{u_{i,t-1}/\sigma_u,\ \cdots ,\ u_{i,0}/\sigma_u,\ w\u{1}_{i,t-1}/\sigma_w ,\ \cdots,\ w\u{1}_{i,0}/\sigma_w}.
\end{align} 
Note that $z_{i,t}$ is a zero-mean subgaussian random vector with variance proxy $1$. 
Now $CA^Lx_{i,t} = CA^L H_t z_{i,t}$, and assuming $p \leq (m+n)t$, $\Fnorm{CA^LH_t}^2 \leq p\opnorm{CA^LH_t}^2 \leq p\opnorm{CA^L}^2 \opnorm{H_t}^2$. 

Note that 
\begin{align}
\opnorm{H_t}^2 &= \opnorm{H_t H_t'} \\
&= \opnorm{\sum_{k=0}^{t-1}\sigma_u^2 A^kB(A^kB)' + \sigma_w^2 A^k (A^k)'}\\
&\leq \opnorm{\Gamma_{\infty}},
\end{align}
so that $\Fnorm{CA^LH_t}^2 \leq p\opnorm{CA^L}^2\opnorm{\Gamma_{\infty}}$. By the Hanson-Wright inequality \cite[Theorem 6.3.2]{vershynin2018high}, 
\begin{align}
	\Pr{\enorm{CA^Lx_{i,t}}^2 \geq (1+\eps) \Fnorm{CA^L H_t}^2}\leq 2\exp\lrp{-c(\min(\eps^2, \eps))p}.
\end{align}
Letting $\eps = c\max(\sqrt{\ln(1/\delta)}, \ln(1/\delta))/p$, and using that $1+\max(\sqrt{\ln(1/\delta)}, \ln(1/\delta)) \leq 2\ln(1/\delta)$ for $\delta \in (0,1/e)$, we have 
\begin{align}
	\Pr{\enorm{CA^Lx_{i,t}}^2 \gsim \opnorm{CA^L}^2 \opnorm{\Gamma_{\infty}} \ln(1/\delta)} \leq 2\delta.
\end{align}
\end{proof}

\subsection{End-to-end analysis}\label{sec:proof:end2end}
\begin{proof} [Proof of Proposition \ref{prop:end2end}]
We trace parameter estimation errors through the four main steps of \cref{alg:sysid}.
\paragraph{Single trajectory estimation}
Given a single trajectory for system $i \in [\Ntot]$, the least squares estimator of the matrix of Markov parameters $G\u{L_1}_i$ is given by
\begin{align}
  \Ghat\u{L_1}_i = \amin_{G \in \R^{p\by (L+1)m}} \Fnorm{Y_i - GU\u{L_1}_i}^2 = Y_i \pseudo{U\u{L_1}_i}.
\end{align}
As an special case of Theorem \ref{thm:MAIN-RESULT} with $N=1$, when the trajectory length $\Tbar$ satisfies \eqref{eq:first-T}, the estimation error of the first $L_1+1$ Markov parameters satisfies
\begin{align}
  \Pr{\opnorm{\Ghat_i\u{L_1} - G_i\u{L_1}} \geq \Delta/(2\sqrt{p})} \leq \delta.
\end{align}
\paragraph{Clustering}
Converting the spectral norm of the estimation error of the first $L_1+1$ Markov parameters to a Frobenius norm bound, we incurring a factor of $\sqrt{p}$ and guarantee that with probability at least $1-\delta$
\begin{align}
   & \Fnorm{\Ghat_i\u{L_1} - G_i\u{L_1}} \leq \sqrt{p} \opnorm{\Ghat_i\u{L_1} - G_i\u{L_1}} \leq \Delta/2.
\end{align}
When the $K$ models in $M$ are separated in the Frobenius norm of the first $(L_1+1)$ Markov parameters of each model by a distance at least $\Delta$, which is twice the estimation error bound, the global minimizer of the sum of squared errors objective of $k$-means, based on the estimated Markov parameters, matches the true cluster assignments of the observed trajectories. As mentioned in the main text, we assume that $k$-means is well-initialized and converges to this minimum, although in general $k$-means is not guaranteed to find the global minimum. In practice, since the clusters are well-separated, running $k$-means multiple times with well-chosen initializations (e.g., $k$-means++) should suffice.

\paragraph{Multi-trajectory estimation}
With correct clustering, we have $N$ trajectories in each cluster.
Under condition \eqref{eq:second-NT} the estimation error for $\Ghat\u{L_2}$ within each cluster is bounded by $\eps$ with probability at least $1-\delta$.
\paragraph{Ho-Kalman realization}
Finally, we transfer the estimation error bounds for Markov parameters to bounds on estimating a realization by using a result on the robustness of the Ho-Kalman algorithm from \cite{oymak2021revisiting}, which is tailored to our setting and stated in Proposition \ref{prop:ho-kalman}. %
Note that this step requires estimating at least $L_2 \geq 2n+1$ impulse response parameters.
\end{proof}
\end{document}